\documentclass[a4paper,UKenglish]{lipics-v2021}

\hideLIPIcs
\nolinenumbers

\usepackage{amsmath}%
\usepackage{amssymb}%
\usepackage{graphicx}%
\usepackage{color}%
\usepackage{xspace}%
\usepackage{mleftright}%
\usepackage{xcolor}%

\usepackage{hyperref}%
\hypersetup{%
   breaklinks,%
   colorlinks=true,%
   linkcolor=[rgb]{0.45,0.0,0.0},%
   citecolor=[rgb]{0,0,0.45}
}

\numberwithin{figure}{section}%
\numberwithin{table}{section}%
\numberwithin{equation}{section}%

    \newtheorem{problem}[theorem]{Problem}%
   
\newcommand{\HLinkShort}[2]{\hyperref[#2]{#1\ref*{#2}}}
\newcommand{\HLink}[2]{\hyperref[#2]{#1~\ref*{#2}}}
\newcommand{\HLinkPage}[2]{\hyperref[#2]{#1~\ref*{#2}%
      $_\text{p\pageref{#2}}$}}
\newcommand{\HLinkPageOnly}[1]{\hyperref[#1]{Page~\refpage*{#1}%
      $_\text{p\pageref{#1}}$}}

\newcommand{\HLinkSuffix}[3]{\hyperref[#2]{#1\ref*{#2}{#3}}}
\newcommand{\HLinkPageSuffix}[3]{\hyperref[#2]{#1\ref*{#2}%
      #3$_\text{p\pageref{#2}}$}}

\newcommand{\figlab}[1]{\label{fig:#1}}
\newcommand{\figref}[1]{\HLink{Figure}{fig:#1}}

\providecommand{\lemlab}[1]{\label{xlemma:#1}}
\renewcommand{\lemlab}[1]{\label{xlemma:#1}}
\newcommand{\lemref}[1]{\HLink{Lemma}{xlemma:#1}}%

\newcommand{\problab}[1]{\label{problem:#1}}
\newcommand{\probref}[1]{\HLink{Problem}{problem:#1}}%

\newcommand{\thmlab}[1]{{\label{theo:#1}}}
\newcommand{\thmref}[1]{\HLink{Theorem}{theo:#1}}

\newcommand{\seclab}[1]{\label{sec:#1}}
\newcommand{\secref}[1]{\HLink{Section}{sec:#1}}

\newcommand{\apndlab}[1]{\label{apnd:#1}}
\newcommand{\apndref}[1]{\HLink{Appendix}{apnd:#1}}

\providecommand{\eqlab}[1]{}%
\renewcommand{\eqlab}[1]{\label{equation:#1}}

\renewcommand{\Re}{\mathbb{R}}%
\newcommand{\eps}{{\varepsilon}}%

\newcommand{\etal}{\textit{et~al.}\xspace}

\DefineNamedColor{named}{RedViolet} {cmyk}{0.07,0.90,0,0.34}

\newcommand{\AlgorithmI}[1]{{\textcolor[named]{RedViolet}{\texttt{\bf{#1}}}}}
\newcommand{\Algorithm}[1]{{\AlgorithmI{#1}\index{algorithm!#1@{\AlgorithmI{#1}}}}}

\newcommand{\cbs}{\mathscr{C}}
\newcommand{\cb}{C}
\newcommand{\bisec}{\beta}
\newcommand{\bisecX}[2]{\beta(#1,#2)}

\newcommand{\distX}[2]{||#1-#2||}

\title{Clustering with Neighborhoods}%

\author{Hongyao Huang}{Department of Computer Science; University of Texas at Dallas; Richardson, TX 75080, USA}{hhuang@utdallas.edu}{}{}
\author{Georgiy Klimenko}{Department of Computer Science; University of Texas at Dallas; Richardson, TX 75080, USA}{gik140030@utdallas.edu}{}{}
\author{Benjamin Raichel}{Department of Computer Science; University of Texas at Dallas; Richardson, TX 75080, USA}{benjamin.raichel@utdallas.edu}{}{}

\authorrunning{H. Huang, G. Klimenko, and B. Raichel}
\Copyright{Hongyao Huang, Georgiy Klimenko, and Benjamin Raichel}



\funding{Partially supported by NSF CAREER Award 1750780.}

\begin{document}

\maketitle

\begin{abstract}
In the standard planar $k$-center clustering problem, one is given a set $P$ of $n$ points in the plane, and the goal is to select $k$ center points, so as to minimize the maximum distance over points in $P$ to their nearest center. Here we initiate the systematic study of the clustering with neighborhoods problem, which generalizes the $k$-center problem to allow the covered objects to be a set of general disjoint convex objects $\cbs$ rather than just a point set $P$. For this problem we first show that there is a PTAS for approximating the number of centers. 
Specifically, if $r_{opt}$ is the optimal radius for $k$ centers, then in $n^{O(1/\eps^2)}$ time we can produce a set of $(1+\eps)k$ centers with radius $\leq r_{opt}$.
If instead one considers the standard goal of approximating the optimal clustering radius, while keeping $k$ as a hard constraint, we show that the radius cannot be approximated within any factor in polynomial time unless $\mathsf{P=NP}$, even when $\cbs$ is a set of line segments. When $\cbs$ is a set of unit disks we show the problem is hard to approximate within a factor of $\frac{\sqrt{13}-\sqrt{3}}{2-\sqrt{3}}\approx 6.99$. This hardness result complements our main result, where we show that when the objects are disks, of possibly differing radii, there is a $(5+2\sqrt{3})\approx 8.46$ approximation algorithm. 
% Additionally, for unit disks we show that there is an FPTAS for the optimal radius when $k$ is bounded.
Additionally, for unit disks we give an $O(n\log k)+(k/\eps)^{O(k)}$  time $(1+\epsilon)$-approximation to the optimal radius, that is, an FPTAS for constant $k$ whose running time depends only linearly on $n$.
Finally, we show that the one dimensional version of the problem, even when intersections are allowed, can be solved exactly in $O(n\log n)$ time.
%
%\keywords{Clustering  \and Hardness \and Approximation.}
\end{abstract}

% \newpage
% \pagenumbering{arabic}

\section{Introduction}
In the standard $k$-center clustering problem, one is given a set $P$ of $n$ points in a metric space and an integer parameter $k\geq 0$, 
%and the goal is to select $k$ points from $P$ (or more generally from the ambient  metric space), 
and the goal is to select $k$ points from the metric space (or from $P$ in the discrete $k$-center problem), 
called centers, so as to minimize the maximum distance over points in $P$ to their nearest center. Equivalently, the problem can be viewed as covering $P$ with $k$ balls with the same radius $r$, where the goal is to minimize $r$. 
It is well known that it is NP-hard to approximate the optimal $k$-center radius $r_{opt}$ within any factor less than $2$ in general metric spaces \cite{hn-ehblp-79}, and that the problem remains hard to approximate within a factor of roughly 1.82 in the plane \cite{fg-oaac-88}.
For general metric spaces, the standard greedy algorithm of Gonzalez \cite{g-cmmid-85}, which repeatedly selects the next center to be the point from $P$ which is furthest from the current set of centers, achieves an optimal $2$-approximation to $r_{opt}$.  
%That is, with $k$ points this algorithms covers all of $P$ with radius $\leq 2r_{opt}$ balls where $r_{opt}$ is the radius of the balls for the optimal solution with $k$ centers. 
An alternative algorithm due to Hochbaum and Shmoys \cite{hs-bphkcp-85} also achieves an optimal approximation ratio of $2$ by approximately searching for the optimal radius, observing that if $r\geq r_{opt}$ then all points will be covered after $k$ rounds of repeatedly removing points in $2r$ radius balls centered at any remaining point of $P$.

In this paper we consider a natural generalization of $k$-center clustering in the plane, where the objects which we must cover are general disjoint convex objects rather than points. Specifically, in the \emph{clustering with neighborhoods} problem the goal is to select $k$ center points so that balls centered at these points with minimum possible radius intersect all the convex objects.
This generalization is natural as real world objects may not be well modeled as individual points. This generalized setting has previously been considered for other classical point based problems in the plane, such as the Traveling Salesperson Problem \cite{dm-aatspn-03}, where the authors referred to these objects as neighborhoods. (We instead typically refer to them as \emph{objects}.) To the best of our knowledge we are the first to consider the general problem of clustering convex objects in this context, though as we discuss below many closely related problems have been considered, some of which equate to special or extreme cases of our problem. We remark that since a point is a convex set, the hardness results for $k$-center clustering immediately apply to clustering with neighborhoods.

\paragraph*{Related Work} As clustering is a fundamental data analysis task, countless variants have been considered. Here we focus on variants which share our $k$-center objective of minimizing the maximum radius of the balls at the chosen centers. Bandyapadhyay \etal \cite{bipv-cack-19} considered the colorful $k$-center problem, where the points are partitioned into color classes $P_1, \ldots, P_c$ and the goal is to find $k$ balls with minimum radius which cover at least $t_i$ points from each color class $P_i$. When our convex objects have bounded diameter our problem can be approximately cast as an instance of colorful $k$-center by replacing each object with the set $P_i$ of grid points it intersects and setting $t_i=1$.
General colorful clustering, however, is more challenging as the color classes can be interspersed, which is why \cite{bipv-cack-19} assumes the number of color classes is a constant, allowing for a constant factor approximation, which subsequently was improved \cite{aakz-takcc-20,jss-fckc-20}.
Note that colorful $k$-center itself  generalizes the $k$-center with outliers problem \cite{ckmn-aflpo-01}, corresponding to the case with a single color class $P$ with  $n-t$ outliers allowed.

Xu and Xu \cite{xx-eaacp-10} considered the $k$-center clustering problem on points sets (KCS) where given points sets $S_1,\ldots,S_n$ the goal is to find $k$ balls of minimum radius such that each $S_i$ is entirely contained in one of the balls.
Again when our objects have bounded diameter we can relate our problem to KCS by discretizing the objects. 
Their requirement that all of $S_i$ be covered by a single ball immediately implies that the optimal radius is at least the radius of the largest object, whereas in our case as only a single point of $S_i$ needs to be covered the radius can be arbitrarily smaller. 
In particular, while \cite{xx-eaacp-10} achieves a $(1+\sqrt{3})$-approximation, we show our problem cannot in general be approximated within any factor in polynomial time unless $\mathsf{P=NP}$.

For the special case when $k=1$ or $k=2$, there are several prior results which closely relate to our problem. 
When $k=1$, i.e.\ the one-center problem, the solution can be derived from the farthest object Voronoi diagram, for which Cheong \etal \cite{c-fpvd-11} gave a near linear time algorithm for polygon objects. 
For disk objects, Ahn \etal \cite{a-cpdtc-13} gave a near quadratic time algorithm for the two-center problem. 
Several papers have also considered generalizing to higher dimensions, but restricting the convex objects to affine subspaces of dimension $\Delta$.
Gao \etal \cite{gls-aidiht-08} introduced the $1$-center problem for $n$ lines, achieving a linear time $(1+\eps)$-approximation, as well as a $(1+\eps)$-approximation for higher dimensional flats or convex sets whose running time depends exponentially on $\Delta$. 
Later in \cite{gls-clhs-10} the same authors considered the more challenging $k=2$ and $k=3$ cases for lines, providing a $(2+\eps)$-approximation in quasi-linear time. Subsequently, \cite{ls-casha-13} considered the problem for axis-parallel flats, where they provide an improved approximation for $k=1$, hardness results for $k=2$, and an approximation for larger $k$ where the time depends exponentially on both $k$ and $\Delta$.
While our focus is on the $k$-center objective, we remark that $k$-means clustering for lines was considered by Marom and Feldman \cite{mf-kclbd-19}, who gave a PTAS for constant $k$.

The $k$-center problem for points in a metric space can also be viewed as clustering the vertices according to the shortest path metric of a positively weighted graph. This allows one to consider specific graph classes, for example, Eisenstat \etal \cite{ekm-akpg-14} gave a polynomial time bi-criteria approximation scheme for $k$-center in planar graphs (i.e.\ they allow both the number of centers and radius to be violated). We remark, however, that for our problem, and the various others described above where the objects are not points, the complete graph with all pairwise distances between the objects, is not necessarily metric (i.e.\ it may not be its own metric completion). For example, the triangle inequality would be violated if you had two small convex objects (e.g.\ points) which are far from one another but both are close to some other large convex object. Note that this non-metric behavior is what allows us to prove a stronger hardness of approximation result than that for points in the plane \cite{fg-oaac-88}. 

Finally, we note that there is a polynomial time algorithm for $k$-center when $k$ is a constant and the objects are points in $d$-dimensional Euclidean space, for constant $d$. Specifically, Agarwal and Procopiuc \cite{ap-eaac-02} gave an $n^{O(k^{1-1/d})}$ time exact algorithm, as well as a $O(n\log k)+(k/\eps)^{O(k^{1-1/d})}$ time $(1+\eps)$-approximation. Later B\u{a}doiu \etal \cite{bhi-accs-02} removed the bounded dimension assumption, achieving a $2^{O((k\log k)/\eps^2)}\cdot dn$ time $(1+\eps)$-approximation.

\paragraph*{Our Contribution}
In this paper we initiate the systematic study of the $\mathsf{NP}$-hard clustering with neighborhoods problem. While this problem allows centers to be placed anywhere in the plane, in \secref{ptas} we first argue that one can compute a cubic sized set of points $P$ and a cubic sized set of radii $R$, such that for any integer $k\geq 0$ there is an optimal set $S\subseteq P$ of $k$ centers with optimal radius $r_{opt}\in R$. 
This naturally leads to a PTAS for approximating the optimal number of centers by using Minkowski sums to reduce the problem to instances of geometric hitting set, for which there is a well known PTAS \cite{mr-irghsp-10}. Specifically, if $r_{opt}$ is the optimal radius for $k$ centers, then in $n^{O(1/\eps^2)}$ time we can produce a set of $(1+\eps)k$ centers with radius $\leq r_{opt}$.

In clustering problems, however, often the emphasis is on approximating the radius, while keeping $k$ as a hard constraint. In \secref{hard} we prove this problem is significantly harder, by adapting the hardness proof of \cite{fg-oaac-88} for planar $k$-center. Specifically, we show that the radius cannot be approximated within any factor in polynomial time unless $\mathsf{P=NP}$, even when the convex objects are restricted to disjoint line segments. On the other hand, for disjoint unit disks, a more in depth proof shows the problem is $\mathsf{APX}$-hard, and in particular cannot be approximated within $\frac{\sqrt{13}-\sqrt{3}}{2-\sqrt{3}}\approx 6.99$ in polynomial time unless $\mathsf{P=NP}$. Complementing this result, in \secref{balls} we present our main result, showing that when the objects are disjoint disks (of possibly varying radii) there is a $(5+2\sqrt{3})$-approximation for the optimal radius. 
Significantly, for the case of disks, our approximation factor of $5+2\sqrt{3} \approx 8.46$ is close to our hardness bound of $\frac{\sqrt{13}-\sqrt{3}}{2-\sqrt{3}}\approx 6.99$. Moreover, while our approximation holds for disks of varying radii, interestingly our hardness bound applies even for disks of uniform radii. 

Further probing the complexity of clustering with neighborhoods, in \secref{bounded} we show there is an FPTAS for unit disks when $k$ is bounded by a constant. Specifically, we give an $O(n\log k)+(k/\eps)^{O(k)}$ time $(1+\eps)$-approximation to the optimal radius, by carefully reducing to the algorithm of \cite{ap-eaac-02} for $k$-center.
Finally in \secref{oned}, by utilizing the searching procedure of \cite{f-pslsct-91}, we show that in one dimension the problem can be solved exactly in $O(n\log n)$ time even when intersections are allowed, contrasting our hardness of approximation results in the plane.

\section{Preliminaries}
Given points $x,y\in \Re^d$, $||x-y||$ denotes their Euclidean distance.  
Given two closed sets $X,Y\subset \Re^d$, $\distX{X}{Y} = \min_{x\in X, y\in Y} ||x-y||$ denotes their distance. For a single point $x$ we write $\distX{x}{Y} = \distX{\{x\}}{Y}$.
For a point $x$ and a value $r\geq 0$, let $B(x,r)$ denote the closed ball centered at $x$ and with radius $r$.

Let $\cbs$ be a set of $n$ pairwise disjoint convex objects in the plane. For simplicity, we assume $\cbs$ is in general position. 
We work under the standard assumption that the objects in $\cbs$ are semi-algebraic sets of constant descriptive complexity. Namely, the boundary of each object is composed of a set of algebraic arcs where the sum of the degrees of these arcs is bounded by a constant, and any natural standard operation on such objects, such as computing the distance between any pair of objects, can be carried out in constant time. See Agarwal \etal \cite{ams-rsss-12} for a more detailed discussion of this model.
Our analysis generalizes to the case where $n$ is the total complexity of $\cbs$ and individual objects in $\cbs$ are not required to have constant complexity, however, assuming constant complexity simplifies certain structural statements and the polynomial degree of $n$ in our running time statements.

\begin{problem}[Clustering with Neighborhoods]\problab{main}
Given a set $\cbs$ of $n$ disjoint convex objects in the plane, and an integer parameter $k\geq 0$, find a set of $k$ points $S$ (called centers) which minimize the maximum distance to a convex object in $\cbs$. That is, 
\[
S=\arg \min_{S'\subset \Re^2, |S'|=k} \max_{\cb\in \cbs} \distX{\cb}{S'}.
\]
%\end{problem}
\end{problem}

Let $S$ be any set of $k$ points, and let $r=\max_{\cb\in \cbs} \distX{\cb}{S}$. We refer to $r$ as the \emph{radius} of the solution $S$, since $r$ is the minimum radius such that the set of all balls $B(s,r)$ for $s\in S$, intersect all $\cb\in \cbs$. If $S$ is an optimal solution then we refer to its radius $r_{opt}$ as the optimal radius. 

In this paper we will consider two types of approximations.%
\footnote{We refrain from using the standard bi-criteria approximation terminology to emphasize that in each case only the size or only the radius is being approximated, not both.
}
Let $\cbs, k$ be an instance  of \probref{main} with optimal radius $r_{opt}$. For a value $\alpha\geq 1$, we refer to a polynomial time algorithm as an \emph{$\alpha$-size-approximation} if it returns a solution $S$ of radius $\leq r_{opt}$ where $|S|\leq \alpha k$. Alternatively, we refer to a polynomial time algorithm as an \emph{$\alpha$-radius-approximation} if it returns a solution $S$ of radius $\leq \alpha r_{opt}$ where $|S|=k$. Often we refer to the latter radius case simply as an $\alpha$-approximation. 

\section{Canonical Sets and a PTAS for Approximating the Size}
\seclab{ptas}

In this section we show that while \probref{main} allows centers to be placed anywhere in the plane, we can compute a canonical cubic sized set of points $P$ and a set of corresponding radii $R$, such that for any integer $k\geq 0$ there is an optimal set $S\subseteq P$ of $k$ centers with optimal radius $r_{opt}\in R$. We then use this property to give a PTAS for \probref{main} when approximating the size of an optimal solution. Specifically, for any fixed $\eps>0$, we give a $(1+\eps)$-size-approximation with running time $n^{O(1/\eps^2)}$. In \secref{balls}, we will again use this canonical set when designing our constant factor radius-approximation for disks.

The \emph{bisector} of two convex objects $\cb,\cb'$ is the set of all points $x$ in the plane such that $\distX{x}{\cb}=\distX{x}{\cb'}$. Let $\bisecX{\cb}{\cb'}$ denote the bisector of $\cb$ and $\cb'$. 
As discussed in \cite{ky-vdpco-03}, any set $\cbs$ of $n$ disjoint constant-complexity convex objects in general position satisfies the conditions of an abstract Voronoi diagram \cite{k-cavd-89}. In particular we can assume the following:
\begin{enumerate}[1)]
 \item For any $\cb, \cb'\in \cbs$ we have that $\bisecX{\cb}{\cb'}$ is an unbounded simple curve.
% \item $T(\cbs)$ is a finite set where $|T(\cbs)| = O(n^3)$.
\item The intersection of any two bisectors is a discrete set with a constant number of points.
\end{enumerate}
We point out that in the following lemma there is a single pair of sets $P,R$ which works simultaneously for all values of $k$. 

\begin{lemma}\lemlab{canon}
 Let $\cbs$ be a set of $n$ disjoint convex objects. In $O(n^3 \log n)$ time one can compute a set of $O(n^3)$ points $P$, and a corresponding set of $O(n^3)$ radii $R$, such that for any value $k\geq 0$ for the instance $\cbs,k$ of \probref{main} there is an optimal set of $k$ centers $S$ with optimal radius $r_{opt}$ such that $S\subseteq P$ and $r_{opt}\in R$. 
\end{lemma}

% \begin{proof}[Proof sketch](For the full proof see \apndref{canonproof})
%  Let $I(\cbs)$ be a set containing exactly one (arbitrary) point from each convex object in $\cbs$. For any number of centers $k$, let $S$ be any optimal solution, and let $r_{opt}$ be the optimal radius. Consider an arbitrary center $s\in S$. Let $\cbs'$ be the subset of objects in $\cbs$ which intersect the ball $B(s,r_{opt})$. We can assume $\cbs'$ is non-empty, as otherwise the center $s$ does not cover any convex object within radius $r_{opt}$ and so can be thrown out. Moreover, if $|\cbs'|=1$ then we can assume $s$ is the point from $I(\cbs)$ which intersects this one convex object. So assume $|\cbs'|>1$. 
%  Let $\cb$ and $\cb'$ be the two furthest objects from $s$ in $\cbs'$, and let $\bisecX{\cb}{\cb'}$ denote their bisector. We argue that without loss of generality $s$ can be assumed to lie on $\bisecX{\cb}{\cb'}$. Moreover, as we move along $\bisecX{\cb}{\cb'}$, $\cb$ and $\cb'$ will remain the furthest objects in $\cbs'$ unless we cross a point which is equidistant to some third object, call this set of points $T_\bisec$. We argue we can further assume $s$ lies either at a point in $T_\bisec$, or at the point minimizing the distance to $\cb$ lying between two consecutive points in $T_\bisec$. Let $M(\cbs)$ be the set of all such points over all bisectors, where by the discussion before the lemma we can argue $|M(\cbs)|=O(n^3)$. Thus the lemma holds as $S$ can be assumed to lie in $I(\cbs)\cup M(\cbs)$.
%  \end{proof}

%\newcommand{\fullproof}{
\begin{proof}
 Let $I(\cbs)$ be a set containing exactly one (arbitrary) point from each convex object in $\cbs$. For any number of centers $k$, let $S$ be any optimal solution, and let $r_{opt}$ be the optimal radius. Consider an arbitrary center $s\in S$. Let $\cbs'$ be the subset of objects in $\cbs$ which intersect the ball $B(s,r_{opt})$. We can assume $\cbs'$ is non-empty, as otherwise the center $s$ does not cover any convex object within radius $r_{opt}$ and so can be thrown out. Moreover, if $|\cbs'|=1$ then we can assume $s$ is the point from $I(\cbs)$ which intersects this one convex object. So assume $|\cbs'|>1$, and let $\cb$ be the convex object in $\cbs'$ which lies furthest from $s$. Now consider moving $s$ continuously toward the convex object $\cb$. As we do so the distance from $s$ to $\cb$ monotonically decreases. Thus so long as $\cb$ remains the furthest convex object from $s$ in $\cbs'$, the ball $B(s,r_{opt})$ still intersects all of $\cbs'$ (i.e.\ we did not increase the solution radius). Now if $\cb$ always remains the furthest, when $s$ eventually reaches and intersects $\cb$ then this will imply its distance to all objects in $\cbs'$ is zero, which is a contradiction as we assumed the convex objects do not intersect.
 Otherwise, at some point $\cb$ is no longer the furthest, which implies we must have crossed a bisector $\bisecX{\cb}{\cb'}$ for some other convex object $\cb'\in \cbs'$.

So far we have shown one can assume each center $s$ either is in  $I(\cbs)$, or lies on the bisector $\bisec=\bisecX{\cb}{\cb'}$ of the two objects, $\cb,\cb'$, which lie furthest away from $s$ among the set of objects $\cbs'$ which intersect the ball $B(s,r_{opt})$.
In the latter case, let $T_\beta$ denote the set of all points $p$ on $\beta$ such that there exists a third object $X\in \cbs$ such that $\distX{p}{X}=\distX{p}{\cb}$ (or equivalently $\distX{p}{X}=\distX{p}{\cb'}$). 
Note that such points lie at intersections of bisectors and thus from the above discussion before the lemma, we know $T_\beta$ is a discrete set. 
As $\beta$ is a simple curve, we can view points in $T_\beta$ as being ordered along $\beta$. 
Suppose that $s\notin T_\beta$, and
let $p$ and $q$ be the points of $T_\beta$ which come immediately before and after $s$ along $\beta$, and let $[p,q]$ denote the portion of $\beta$ lying between these points. 
(This interval may be unbounded to one side if $s$ comes after or before all points in $T_\beta$.)
Recall that $\cbs'$ is the subset of objects intersecting $B(s,r_{opt})$, and $\cb$ and $\cb'$ are the furthest from $s$ among those in $\cbs'$. 
Observe that for any other point $z$ in $[p,q]$, $\cb$ and $\cb'$ must also be the furthest objects from $z$ among those in $\cbs'$, as otherwise as we move continuously along $\beta$ from $s$ to $z$ we must cross another point from $T_\beta$ before reaching $z$ and there are no such points in $(p,q)$.
Thus if we replace $s$ with the point in $[p,q]$ minimizing the distance to $\cb$ (or equivalently $\cb'$) then all objects previously intersected by $B(s,r_{opt})$ will remain intersected by $B(s,r_{opt})$. 
 
Let $M(\cbs)$ be a set containing, for each bisector $\beta$, the set $T_\beta$ and one minimum distance point from each such interval $[p,q]$. 
We thus have argued that the points of $S$ can be assumed to lie in $P=I(\cbs)\cup M(\cbs)$.
 As for the running time and size of these sets, first observe that $I(\cbs)$ has size $n$ and can be trivially computed in $O(n)$ time. For the set $M(\cbs)$, first observe that there are $O(n^2)$ bisectors. 
 For any bisector $\bisec$, the set $T_\beta$ of intersection points of $\bisec$ with other bisectors that are equidistant at the intersection point, has size $O(n)$, since by general position every point is equidistant to at most 3 objects and as mentioned above any pair of bisectors intersect in a constant number of points.
 (In other words, we ultimately consider all $O(n^3)$ points equidistant to three objects, as opposed to all $O(n^4)$ bisector intersections.) 
 Thus the set $M(\cbs)$, and correspondingly $P$, has size $O(n^3)$ as claimed. For the running time, as the objects in $\cbs$ all have constant complexity, so do their bisectors, and thus $T_\beta$ can be computed in $O(n)$ time. The minimum points of $M(\cbs)$ on $\beta$, can thus be computed by sorting $T_\beta$ along $\beta$, in $O(n \log n)$ time, and then computing the minimum point in constant time for each constant complexity interval between consecutive pairs of points from $T_\beta$ along $\beta$. Thus over all $O(n^2)$ bisectors it takes $O(n^3 \log n)$ time to compute $M(\cbs)$.
\end{proof}
%}
  
We now argue the canonical sets $P$ and $R$ from the above lemma naturally lead to a PTAS for size-approximation by using Minkowski sums.
For sets $A,B \subset \Re^2$, let $A\oplus B = \{a+b\mid a\in A, ~b\in B\}$ denote their Minkowski sum. Let $B(r)$ denote the ball of radius $r$ centered at the origin. Then we write $\cbs\oplus B(r) = \{\cb\oplus B(r) \mid \cb\in \cbs\}$. A set of points $S$ is called a \emph{hitting set} for a set of objects if every object has non-empty intersection with $S$.

\begin{observation}
 A set $S$ of $k$ centers is a solution to \probref{main} of radius $r$ if and only if $S$ is a hitting set of size $k$ for $\cbs\oplus B(r)$. 
 This holds since for any $\cb\in \cbs$ and $s\in S$,  $B(s,r)\cap \cb\neq \emptyset$ if and only if $s\in \cb\oplus B(r)$.
\end{observation}

In the geometric hitting set problem we are given a set $\mathcal{R}$ of $n$ regions and a set $P$ of $m$ points in the plane, and the goal is to select a minimum sized hitting set for $\mathcal{R}$ using points from $P$.
The above observation implies we can reduce any given instance $\cbs, k$ of \probref{main} to multiple instances of geometric hitting set. Specifically, by \lemref{canon}, in $O(n^3\log n)$ time we can compute a set $R$ of $O(n^3)$ values, one of which must be the optimal radius $r_{opt}$. Then for each $r\in R$ we construct a hitting set instance where $\mathcal{R}=\cbs\oplus B(r)$, and $P$ is the set of points from \lemref{canon}.
By the above observation, if $r<r_{opt}$, then the hitting set instance requires more than $k$ points, and if $r\geq r_{opt}$ then it requires at most $k$ points. Therefore, given an algorithm for geometric hitting set we can use it to binary search for $r_{opt}$. 

While hitting set is in general $\mathsf{NP}$-hard to approximate within  logarithmic factors \cite{rs-scep-97}, in our case there is a PTAS as the regions are nicely behaved. A collection of regions in the plane is called a set of \emph{pseudo-disks} if the boundaries of any two distinct regions in the set cross at most twice. Mustafa and Ray \cite{mr-irghsp-10} showed that there is an $nm^{O(1/\eps^2)}$ time PTAS for geometric hitting set when $\mathcal{R}$ is a collection of $n$ pseudo-disks and $P$ is a set of $m$ points.
It is known that if we take the Minkowski sum of a single convex object with each member of a set of disjoint convex objects, then the resulting set is a collection of pseudo-disks (see for example \cite{aps-su-08}). Thus $\cbs\oplus B(r)$ is a collection of pseudo-disks. 
Therefore, by the above discussion, we have the following theorem.
%Note that 
As the decision procedure is now approximate, the binary search must be modified to look at larger radii when the hitting set algorithm returns $>(1+\eps)k$ points, and smaller radii otherwise. 
%(For a more detailed example of binary searching with an approximate decider see the proof of \thmref{radiusapprox}.)
(This yields an adjacent pair $r<r'$ such that $r<r_{opt}$, implying $r'\leq r_{opt}$, and an $r'$- cover of the input using $\leq (1+\eps)k$ points.)

\begin{theorem}
There is a PTAS for \probref{main} for approximating the optimal solution size. That is, for any fixed $\eps>0$, there is a $(1+\eps)$-size-approximation with running time $n^{O(1/\eps^2)}$.
\end{theorem}

We remark that the PTAS of \cite{mr-irghsp-10} implicitly assumes the objects are in general position, that is if two objects intersect then they properly intersect (i.e.\ their interiors intersect). While $\cbs$ satisfies this property, it may not after we take the Minkowski sum with a given radius. However, as we can compute distances between our objects, this is easily overcome by computing the smallest non-zero distance $d$ between two objects in $\cb\oplus B(r)$, and instead running the hitting set algorithm on $\cb\oplus B(r+\alpha)$, where $\alpha$ is some infinitesimal value less than $d/2$. This ensures any objects which intersected in $\cb\oplus B(r)$ now properly intersect, and there are no new intersections. 

%%%%%%%%%%%%%%%%%%%%%%%%%%%%%%%%%%
%%%%%%%%%%%%%%%%%%%%%%%%%%%%%%%%%%

\section{Radius Approximation Hardness}
\seclab{hard}
In this section we argue that for \probref{main} it is hard to approximate the radius within any factor, even when $\cbs$ is restricted to being a set of line segments.  Moreover, for the case when $\cbs$ is a set of disks, i.e.\ the case considered in \secref{balls}, we argue the problem is $\mathsf{APX}$-Hard. Our hardness results use a construction similar to the one from \cite{fg-oaac-88}, where they reduce from the problem of planar vertex cover where the maximum degree of a vertex is three, which is known to be $\mathsf{NP}$-complete \cite{vc3}. We denote this problem as $\mathsf{P3VC}$.

\subsection{Line Segments}

Here we argue that it is hard to radius-approximate \probref{main} within any factor, even when $\cbs$ is a set of line segments.
We remark that the following reduction works for any instance of planar vertex cover (i.e.\ regardless of the degree), but the reduction for disks in the next subsection uses that the degree is at most three. 

\begin{theorem}\thmlab{seghard}
% \probref{main} cannot be radius-approximated within any factor in polynomial time unless $\mathsf{P}=\mathsf{NP}$, even when restricting to the set of instances in which $\cbs$ is a set of disjoint line segments.
\probref{main} cannot in polynomial time be radius-approximated within any factor that is computable in polynomial time unless $\mathsf{P}=\mathsf{NP}$, even when restricting to the set of instances in which $\cbs$ is a set of disjoint line segments.
\end{theorem}
\begin{figure}[t]
\centering
\includegraphics[scale=.6]{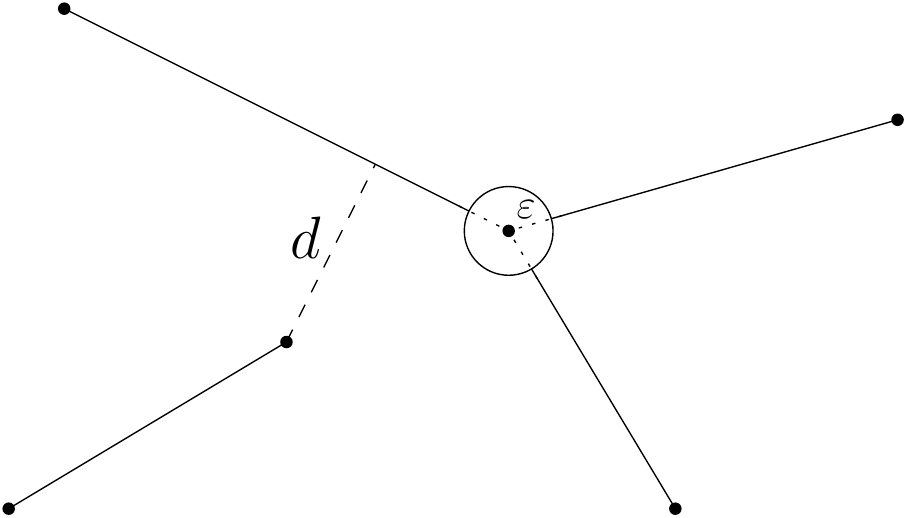}
\caption{Reducing planar Vertex Cover to \probref{main} for segments.}
\figlab{segments}
\end{figure}
\begin{proof}
Let $G,k$ be an instance of $\mathsf{P3VC}$. Consider a straight line embedding of $G$, and let $d$ denote the distance between the closest pair of non-adjacent segment edges.% 
\footnote{
In $O(n\log n)$ time one can compute a straight line embedding of $G$ where the vertices are on an $(2n-4) \times (n-2)$ grid \cite{dpp-hdpgg-90}. This implies a lower bound on $d$ with a polynomial number of bits.}
Let $\eps>0$ be a value strictly smaller than $d/2$ and strictly smaller than half the length of any segment edge. The set $\cbs$ of segments in our instance of \probref{main} will be the segment edges from the embedding, but where each segment has an $\eps$ amount removed from each end, i.e.\ we remove all portions of segments in $\eps$ balls around the vertices, see \figref{segments}. We use the same value of $k$ in our \probref{main} instance as in the $\mathsf{P3VC}$ instance.

If there is a vertex cover of size at most $k$ then if we place balls of radius $\eps$ at each of the $k$ corresponding vertices of the embedding, then these balls will intersect all segments in $\cbs$, i.e.\ we have a solution to \probref{main} of radius $\eps$. On the other hand, by the definition of $d$, any ball of radius $<d/2$ cannot simultaneously intersect two segments from $\cbs$ if they correspond to non-adjacent edges from $G$. (Note when we shrunk the edges by $\eps$ this could only have made them further apart.) Thus if the minimum vertex cover requires $>k$ vertices, then our instance of \probref{main} requires $>k$ centers if we limit to balls with radius $<d/2$.  

Therefore, if we could approximate the minimum radius of our \probref{main} instance within any factor less than $d/(2\eps)$ then we can determine whether the corresponding vertex cover instance had a solution with $\leq k$ vertices. However, we are free to make $\eps>0$ as small as we want and thus $d/(2\eps)$ as large as we want, so long as this quantity (or more precisely a lower bound on it) is computable in polynomial time.
\end{proof}

\subsection{Disks}

Here we argue that it is hard to radius-approximate \probref{main} within a constant factor when $\cbs$ is restricted to be a set of unit disks.
The following reduction from $\mathsf{P3VC}$ is similar to the one given in \cite{fg-oaac-88}, which embeds the graph such that edges are replaced by odd length sequences of points. In our case, these odd length sequences of points are instead replaced with odd length sequences of appropriately spaced disks.

\begin{theorem}\thmlab{hardmain}
For the set of instances in which $\cbs$ is a set of disjoint unit disks, 
\probref{main} cannot be radius-approximated to any factor less than $\frac{\sqrt{13} - \sqrt{3}}{2-\sqrt{3}}$  in polynomial time unless $\mathsf{P=NP}$.
\end{theorem}
\begin{proof} %[Proof (reduction only)]
% Here we describe the reduction from $\mathsf{P3VC}$ to \probref{main}. In \apndref{hardnesscontinued} we argue why this reduction is valid and achieves the stated approximation hardness factor, along with figures from the reduction description.
%
To simplify our construction description, instead of requiring the disks be disjoint, we allow them to intersect at their boundaries, but not their interiors. Later we remark how this easily implies the result for the disjoint disk case.

So let $G=(V,E),k$ be an instance of $\mathsf{P3VC}$. For every edge in $E$ we create a sequence of an odd number (greater than 1) of unit disks, where consecutive disks in the sequence are spaced $2(2/\sqrt{3}-1)$ apart from one another. 
(Note $2(2/\sqrt{3}-1)$ is the distance between the disks, not their centers.)
For a vertex $v$ of degree two, we place the disks corresponding to the $v$ end of the adjacent edges again at distance $2(2/\sqrt{3}-1)$ apart. 
For a vertex $v$ of degree three, we place the disks corresponding to the $v$ end of the adjacent edges such that they all just touch one another at their boundaries, see \figref{threecircles}. Thus the centers of these disks form an equilateral triangle, and let the center point of this triangle be $t$. For any one of the adjacent edges, we further require that the centers of the first two disks (on the $v$ end of the edge) lie on a straight line containing $t$, in other words the edges leaving $v$ do not bend until several disks away from $v$. As $G$ is a planar graph with maximum degree three such an embedding of polynomial size is possible, similar to the case in \cite{fg-oaac-88}. Doing so requires using different numbers of disks for each edge and allowing the edges to bend (i.e. the centers of three consecutive disks of an edge may not lie on a line). However, we will require these bends to be gradual. Specifically, observe that if the centers of three consecutive disks of an edge were on a straight line, the distance between the two non-consecutive disks would be $2(1+2(2/\sqrt{3}-1))>2.6$, see \figref{twocircles}. We then require that the bends are shallow enough such that two non-consecutive disks of an edge are more than $2.5$ apart. We also require this for disks from edges adjacent to a degree two vertex (when they are not both the disks immediately adjacent at the vertex), or a degree three vertex when neither disk is one of corresponding three touching disks of the vertex. Finally, for disks that come from edges that are not adjacent, we easily enforce that they are again more than 2.5 apart. (This is similar to the value $d$ from \thmref{seghard}.)

\begin{figure}[t]
\centering
\includegraphics[scale=.5]{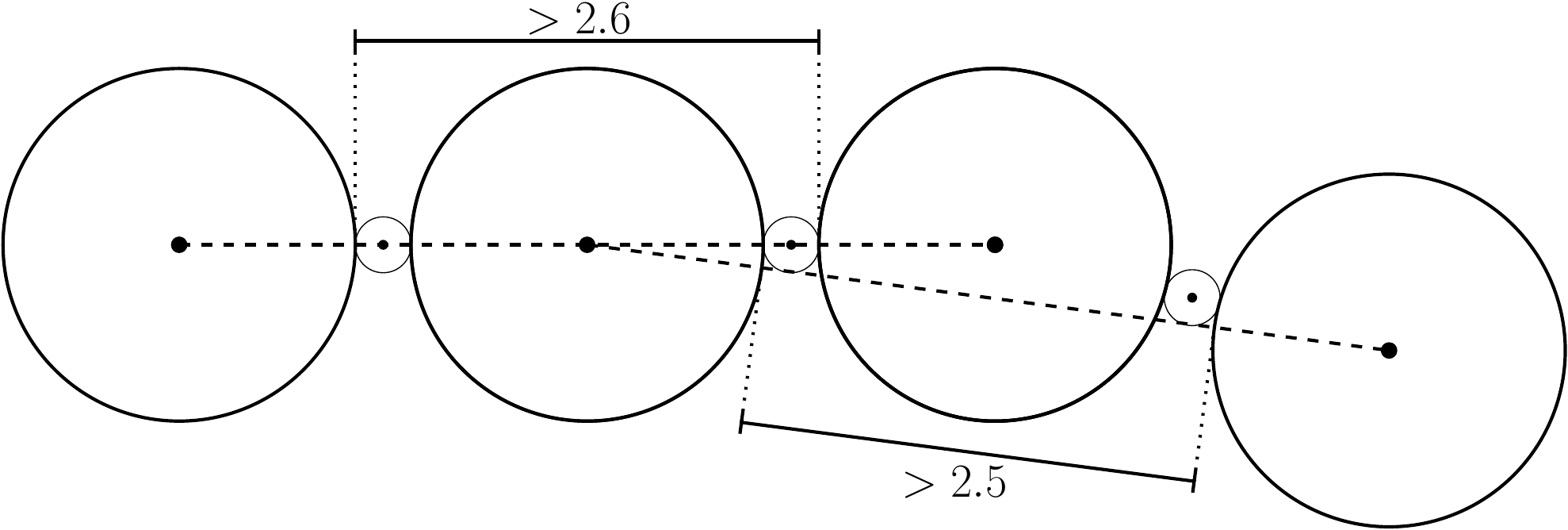}
\caption{Consecutive disks along an edge.}
\figlab{twocircles}
\end{figure}

So given an instance $G,k$ of $\mathsf{P3VC}$, we construct an instance $\cbs, \kappa$ of \probref{main} where $\cbs$ is determined from $G$ as described above and $\kappa=k+(|\cbs|-|E|)/2$. 
We first argue if $G$ has a vertex cover of size $k$ then for our instance of \probref{main} there is a solution of radius $2/\sqrt{3}-1$. First, for any vertex $v$ in the vertex cover we create a center, and roughly speaking place it at the location of $v$ in the embedding. Namely, if $v$ had degree two then we place the center at the midpoint of the centers of the disks at the ends of the edges adjacent to $v$, which by construction are exactly $2(2/\sqrt{3}-1)$ apart and thus a ball at the midpoint with radius $2/\sqrt{3}-1$ intersects both. If $v$ has degree three then we place a center at the center point $t$ of the equilateral triangle determined by three touching disks of the adjacent edges. An easy calculation\footnote{For an equilateral triangle with edge length 2, the distance from an edge to the center point of the triangle is $1/\sqrt{3}$, thus the distance from the center point to any one of the unit balls is $2/\sqrt{3}-1$.} shows that since our disks have unit radius, that $B(t, 2/\sqrt{3}-1)$ intersects the three touching disks. We now cover the remaining disks with $(|\cbs|-|E|)/2$ centers. For any edge $e\in E$ let $n_e$ be the number of disks used for $e$ in the above construction. Observe that as we already placed centers at vertices corresponding to a vertex cover of the edges, at least one disk at the end of each edge is already covered, and so there are at most $n_e-1$ consecutive disks that need to be covered. (Note $n_e-1$ is even.)  However, as consecutive disks are exactly $2(2/\sqrt{3}-1)$ apart on each edge, these $n_e-1$ disks can be covered with $(n_e-1)/2$ balls of radius $(2/\sqrt{3}-1)$ by covering the disks in pairs. Thus the total number of centers used is 
$k+\sum_{e\in E} (n_e-1)/2 = k+(|\cbs|-|E|)/2=\kappa$.
 
\begin{figure}[t]
\centering
\includegraphics[scale=.5]{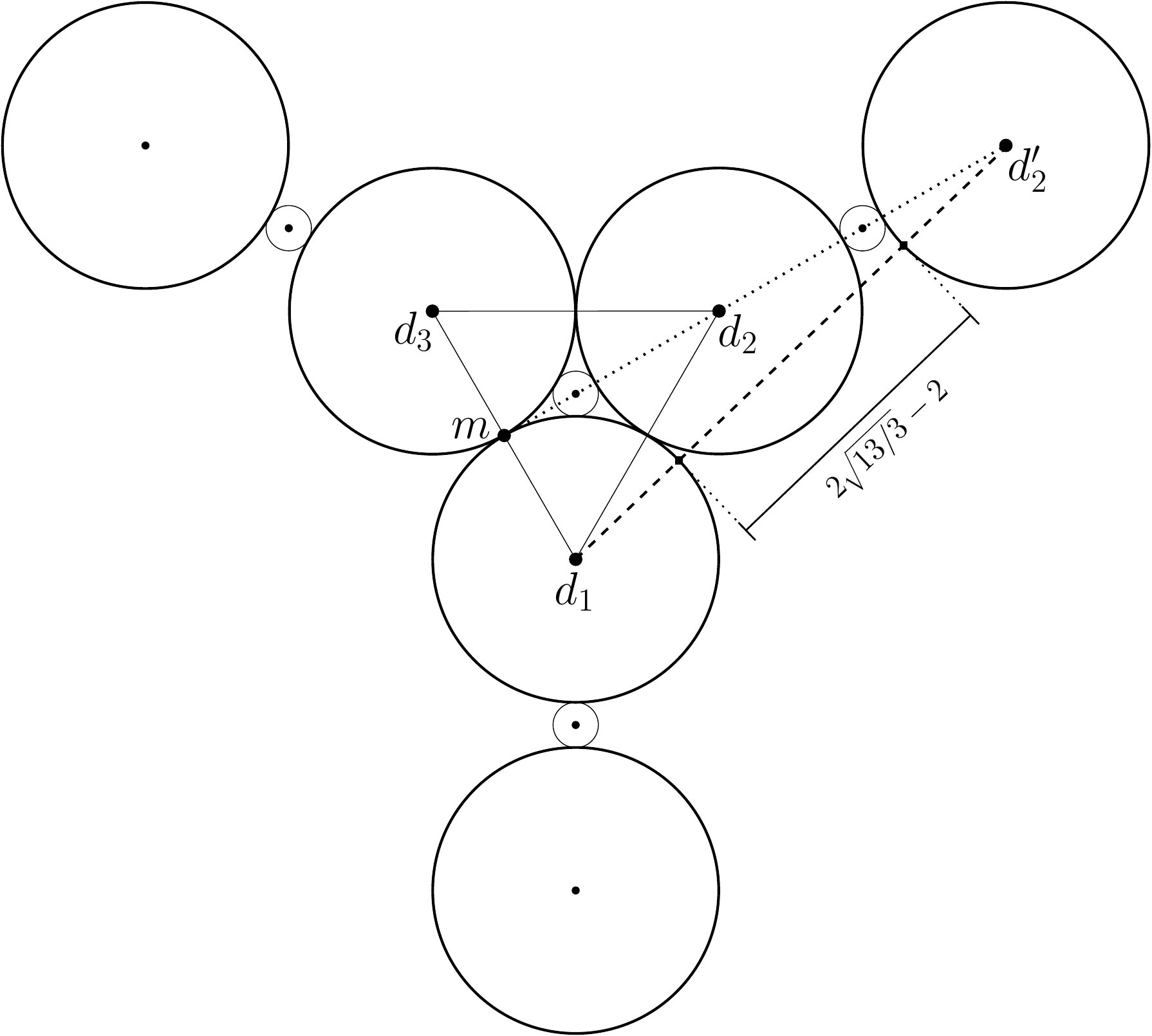}
\caption{Three touching disks corresponding to a degree three vertex.}
\figlab{threecircles}
\end{figure}

Now suppose the minimum vertex cover of $G$ requires $>k$  vertices. In this case we argue that our instance of \probref{main} requires more than $\kappa$ centers if we limit to balls with radius $<\sqrt{13/3}-1$. Call any two disks in $\cbs$ neighboring if they are consecutive on an edge or if they are disks on the $v$ end of two edges adjacent to a vertex $v$. By construction, neighboring disks have distance $\leq 2/\sqrt{3}-1$ from each other. For a pair of disks which are not neighboring we now argue their distance is at least $2\sqrt{13/3} -2$. Specifically, if these disks come from the same edge but are not consecutive along that edge, or if they are from distinct edges that are either non-adjacent or are adjacent to a degree two vertex (but not the two disks of that vertex), then by construction their distance is $>2.5>2\sqrt{13/3} -2$. The remaining case is when the disks are from distinct edges adjacent to a degree three vertex, but they are not both from the three touching disks of the vertex. It is easy to see that the closest two such disks can be is when one of the disks is one of the three touching disks, and the other is the second disk on another edge. We now calculate the distance between two such disks, see \figref{threecircles}. Let the three touching disks be denoted $D_1$, $D_2$, and $D_3$, with centers $d_1$, $d_2$, and $d_3$, respectively. Let $D_2'$ denote the second disk on the edge containing $D_2$, and let its center by $d_2'$. We wish to compute $\distX{D_1}{D_2'} =\distX{d_1}{d_2'}-2$, as these are unit disks. Let $m$ denote the midpoint of $d_1$ and $d_3$, and observe that the line through $d_2$ and $d_2'$ passes through $m$ and is orthogonal to the line through $d_1$ and $d_3$, as the points $d_1$, $d_2$, and $d_3$ form and equilateral triangle. Thus by the Pythagorean theorem we have $\distX{d_1}{d_2'}^2 = 1^2+(1+2(2/\sqrt{3}-1)+1+\sqrt{3})^2 = 1+(4/\sqrt{3}+\sqrt{3})^2=52/3$, where the $+\sqrt{3}$ term is the height of an equilateral triangle of side length $2$. Thus $\distX{D_1}{D_2'} =\distX{d_1}{d_2'}-2 = 2\sqrt{13/3} -2$.

Now we finish the argument that when the minimum vertex cover of $G$ requires $>k$  vertices, our instance of \probref{main} requires more than $\kappa$ centers if we limit to balls with radius $<\sqrt{13/3}-1$. By the above, limiting to radius $<\sqrt{13/3}-1$ implies that any ball either covers just a single disk, or a pair of neighboring disks. An edge $e$ with $n_e$ disks thus requires at least $\lceil n_e/2 \rceil = 1+(n_e-1)/2$ disks to cover it. Moreover, a ball can only cover both a disk of $e$ and $e'$ if those disks are on the $v$ end of two edges adjacent to $v$. Let $E_z$ be the subset of edges with at least one disk covered by such a ball (i.e.\ a ball corresponding to a vertex), and let $z$ be the number of such balls. Then the total number of balls required is 
\begin{align*}
&\geq z+\sum_{e\in E_z} (n_e-1)/2 + \sum_{e\in E\setminus E_z} (1+(n_e-1)/2)\\ 
&= z + (|\cbs|-|E|)/2 + |E\setminus E_z| 
= z+(\kappa-k) + |E\setminus E_z|, 
\end{align*}
which is more than $\kappa$ when $z+|E\setminus E_z|> k$. 
Notice, however, there is a vertex cover of $G$ of size $z+|E\setminus E_z|$, consisting of the vertices that $z$ counted, and one vertex from either end of each edge in $E\setminus E_z$. Thus as the minimum vertex cover has size $>k$, we have $z+|E\setminus E_z|>k$ as desired. 

Therefore, if we could approximate the minimum radius of our \probref{main} instance within any factor less than $\frac{\sqrt{13/3}-1}{2/\sqrt{3}-1} = \frac{\sqrt{13} - \sqrt{3}}{2-\sqrt{3}}$ then we can determine whether the corresponding vertex cover instance had a solution with $\leq k$ vertices. In the above analysis the boundaries of the circles were allowed to intersect, but we can enforce that all disks are disjoint without changing the approximation hardness factor since we showed the problem is hard for any factor that is less than $\frac{\sqrt{13} - \sqrt{3}}{2-\sqrt{3}}$. Specifically, rather than having the disks for a degree three vertex touch, we can instead make them arbitrarily close to touching. 
\end{proof}

%%%%%%%%%%%%%%%%%%%%%%%%%%%%%%%
%%%%%%%%%%%%%%%%%%%%%%%%%%%%%%%

\section{Constant Factor Radius Approximation for Disks}
\seclab{balls}

In this section we argue that when $\cbs$ is a set of disjoint disks (of possibly differing radii), that there is a constant factor radius-approximation for \probref{main}. 
%Our algorithm makes use of the following critical lemma.

\begin{lemma}%[Proof in \apndref{constantappendix}]
\lemlab{arc}
Let $\cbs$ be a set of pairwise disjoint disks such that for all $\cb\in \cbs$, the radius of $\cb$ is $\geq r$. If there is a point $s\in \Re^2$ where $\distX{s}{\cb}\leq (2/\sqrt{3}-1)r$ for all $\cb\in \cbs$, then $|\cbs|\leq 2$.
\end{lemma}
\newcommand{\twoproof}{
\begin{proof}
We give a proof by contradiction. So suppose there exists a point $s$ such that there are three disjoint disks in $\cbs$, each with radius $\geq r$, and all of which intersect the ball $B(s,(\frac{2}{\sqrt{3}}-1)r)$. Observe that if any one of these three disjoint disks $\cb$ has radius $>r$, then it can be replaced by a disk $\cb'$ of radius $r$ such that $\cb'\subset \cb$ and $\cb'$ still intersects $B(s,(\frac{2}{\sqrt{3}}-1)r)$. As these new disks are all still disjoint and intersect $B(s,(\frac{2}{\sqrt{3}}-1)r)$, it suffices to argue we get a contradiction when all three disks have radius exactly $r$.  Let the centers of these three disks be denoted $x$, $y$, and $z$. 
Now, at least one of the angles $\angle xsy$, $\angle ysz$, and $\angle zsx$ is $\leq 2\pi/3$. Without loss of generality assume it is $\angle xsy$, and let $\gamma=\angle xsy$.

Consider the triangle $\triangle sxy$, and let its side lengths be denoted $a=\distX{x}{s}, b=\distX{y}{s}, c=\distX{x}{y}$.
Since $\gamma\leq 2\pi/3$, by the Law of Cosines we thus have $c^2 = a^2 + b^2 - 2ab \cos(\gamma) \leq a^2 + b^2 + ab$.
As the $r$ radius disks with centers $x$ and $y$ are disjoint, we know that $2r<c$. Combining these two inequalities we get
$4r^2 < a^2 + b^2 + ab$.
As $B(s,(\frac{2}{\sqrt{3}}-1)r)$ intersects the $r$ radius disks centered at both $x$ and at $y$,
%As the segments $\overline{xs}$ and $\overline{ys}$ both intersect $B(s,(\frac{2}{\sqrt{3}}-1)r)$, 
we also have that $a,b\leq (\frac{2}{\sqrt{3}}-1)r + r = \frac{2r}{\sqrt{3}}$. Combining this with the previous inequality gives $4r^2 < a^2 + b^2 + ab\leq 4r^2/3+4r^2/3+4r^2/3 = 4r^2$, which is a clear contradiction and thus the number of disks in $\cbs$ is at most 2.
\end{proof}
}
\twoproof

For any constant $c\geq 1$, we call an algorithm a \emph{$c$-decider} for \probref{main}, if for a given instance with optimal radius $r_{opt}$, and for any given query radius $r$, 
if $r\geq r_{opt}$ then the algorithm returns a solution $S$ of radius $\leq cr$, and if $r< r_{opt}/c$ it returns False (for $r_{opt}/c\leq r<r_{opt}$ either answer is allowed).

\begin{lemma}\lemlab{decider}
 There is an $O(n^{2.5})$ time $(5+2\sqrt{3})$-decider for \probref{main}, when restricted to instances where $\cbs$ is a set of disjoint disks.
\end{lemma}

\begin{proof}
Let $r$ be the given query radius. 
We build a set $S$ of centers as follows, where initially $S=\emptyset$.
Let $P$ be the set of center points of all disks in $\cbs$ with radius $<(3+2\sqrt{3}) r$. 
Until $P$ is empty repeatedly add an arbitrary point $p\in P$ to the set $S$, remove all disks from $\cbs$ which intersect $B(p,(5+2\sqrt{3})r)$, and remove all center points from $P$ corresponding to disks removed from $\cbs$. Let $S_1$ refer to the resulting set of centers.
For the remaining set of disks $\cbs'$, define the subset $\cbs'' = \{\cb\in \cbs' \mid \exists~ D\in \cbs'\setminus\{\cb\} \text{ s.t. } \distX{\cb}{D}\leq 2r\}$. First, for every disk $\cb$ in $\cbs'\setminus \cbs''$ we add an arbitrary point from $\cb$ to $S$. Let this set of added centers be denoted $S_2$. Now for the set $\cbs''$ we construct a graph $G=(V,E)$ where $V=\cbs''$ and there is an edge from $\cb$ to $D$ if and only if $\distX{\cb}{D}\leq 2r$. Let $\mathcal{E}$ be a minimum edge cover of $G$. 
(Note every vertex in $G$ has an adjacent edge by the definition of $\cbs''$ and thus $\mathcal{E}$ exists.)
For every edge $(\cb,D)\in \mathcal{E}$, $\distX{\cb}{D}\leq 2r$ and thus there is a point $p\in \Re^2$ such that $\distX{p}{\cb}, \distX{p}{D}\leq r$. So finally, for each $(\cb,D)\in \mathcal{E}$ we add this corresponding point $p$ to $S$. Let this final set of added centers be denoted $S_3$. If $|S|\leq k$ we return $S$ (which is the disjoint union of $S_1$, $S_2$, and $S_3$) and otherwise we return False.

To prove the above algorithm is a $(5+2\sqrt{3})$-decider, first we argue that if $r<r_{opt}/(5+2\sqrt{3})$ then it returns False. To do so we prove the contrapositive. So assume $|S|\leq k$. Let $S_1$, $S_2$, and $S_3$, and $\cbs'' \subseteq \cbs' \subseteq \cbs$ be as defined above. As we used balls of radius $(5+2\sqrt{3})r$, all $\cb\in \cbs\setminus \cbs'$ are within distance $(5+2\sqrt{3})r$ of points in $S_1$. All $\cb \in \cbs'\setminus \cbs''$ have distance zero to a point in $S_2$. Finally, all $\cb \in \cbs''$ have distance $\leq r$ to a point in $S_3$. 
As $S$ is the disjoint union of $S_1$, $S_2$, and $S_3$, we thus have that all $\cb\in \cbs$ are within distance $(5+2\sqrt{3})r$ to a set $S$ with $\leq k$ points, which by the definition of \probref{main} means that $r_{opt}\leq (5+2\sqrt{3})r$.

Now suppose $r\geq r_{opt}$, where $r_{opt}$ is the optimal radius for the given instance $\cbs, k$ of \probref{main}.
In order to prove the algorithm is a $(5+2\sqrt{3})$-decider, in this case we must argue it returns a $\leq (5+2\sqrt{3})r$ radius solution. As already shown above, if the algorithm returns a solution then it has radius $\leq (5+2\sqrt{3})r$, thus all we must argue is that a solution is returned, namely that $|S|\leq k$. So fix an optimal solution $S^*$ for the original input instance $\cbs, k$. 
We argue that there are disjoint subsets $S^*_1$, $S^*_2$, and $S^*_3$ of $S^*$ such that $|S^*_1|\geq |S_1|$, $|S^*_2|\geq |S_2|$, and $|S^*_3|\geq |S_3|$, and therefore $|S|\leq |S^*|=k$.

Let the points in $S_1=\{t_1,\ldots, t_{|S_1|}\}$ be indexed in the order they were selected. Consider the point $t_i$, which is the center of some disk $\cb_i\in \cbs$ with radius $\leq (3+2\sqrt{3})r$.
Let $U_i = \cup_{j\leq i} B(t_j,(5+2\sqrt{3})r)$.
Define $S^*_1$ as the centers $s$ of $S^*$ such that $B(s,r)\subseteq U_{|S_1|}$.
To argue $|S_1^*|\geq |S_1|$, it suffices to argue that for all $i$ there exists some $s\in S^*$ such that $B(s,r)\not \subseteq U_{i-1}$ while $B(s,r)\subseteq U_{i}$ (i.e.\ $s$ gets charged uniquely to $t_i$).
Now there must be some center $s\in S^*$ such that $B(s,r_{opt})\cap \cb_i\neq \emptyset$, as $S^*$ covers $\cbs$ with radius $r_{opt}$. Moreover, since $r\geq r_{opt}$, we have $B(s,r)\not \subseteq U_{i-1}$, since otherwise it implies $U_{i-1}\cap \cb_i\neq \emptyset$ and thus $t_i$ could not have been selected in the $i$th round as the algorithm had already removed it from $P$. Conversely, $B(s,r)\subseteq U_{i}$, since $B(s,r)$ intersects $\cb_i$ and $\cb_i$ has radius $\leq (3+2\sqrt{3})r$, and thus $B(s,r)\subseteq B(t_i,(5+2\sqrt{3})r)\subseteq U_{i}$. Therefore $|S_1^*|\geq |S_1|$.

For any $s\in S^*_1$, $B(s,r)\subseteq U_{|S_1|}$, and since the disks of $\cbs'$ do not intersect $U_{|S_1|}$, in the optimal solution $\cbs'$ must be $r_{opt}$-covered only using centers from $S^*\setminus S^*_1$.
Let $S^*_2$ be the subset of centers from $S^*\setminus S^*_1$ which $r_{opt}$ covers $\cbs'\setminus \cbs''$.
Since any disk $\cb\in (\cbs'\setminus \cbs'')$  has distance $>2r$ to its nearest neighbor in $\cbs'\setminus \{\cb\}$ and $r_{opt}\leq r$, the optimal solution must use a distinct center to cover each disk in $\cbs'\setminus \cbs''$, i.e.\ $|S^*_2|\geq |S_2|$, and moreover, $\cbs''$ must be covered in the optimal solution by $S^*\setminus (S^*_1\cup S^*_2)$. So finally, let $S^*_3$ be the subset of centers from $S^*\setminus (S^*_1\cup S^*_2)$ which $r_{opt}$ covers $\cbs''$. By construction, the radius of each $\cb\in \cbs''$ is $\geq (3+2\sqrt{3})r$. Thus, by \lemref{arc} any point from $S^*_3$ can $(2/\sqrt{3}-1)\cdot (3+2\sqrt{3})r = r\geq r_{opt}$ cover at most 2 disks from $\cbs''$. Now the graph $G$, for which our algorithm computes a minimum edge cover $\mathcal{E}$, contains an edge for every pair of disks which can be simultaneously covered with a single $r$ radius ball. Therefore $|S^*_3|\geq |\mathcal{E}| = |S_3|$.

For the running time, computing the set $P$ takes $O(n)$ time. Selecting a new point $p\in P$ and removing all disks from $\cbs$ which intersect $B(p,(5+2\sqrt{3})r)$ can be done in $O(n)$ time, and thus repeating this till $P$ is empty takes $O(n^2)$ time. Determining the subset $\cbs''$, and hence the graph $G$, can naively be done in $O(n^2)$ by checking the distances between all pairs in $\cbs'$. Selecting a point from each $\cb\in (\cbs'\setminus \cbs'')$ takes $O(n)$ time. Finally, since computing a minimum edge cover can be reduced to computing a maximum matching, $\mathcal{E}$ can be found in $O(n^{2.5})$ time (see \cite{mv-afmmgg-80}). 
\end{proof}

We remark that it should be possible to improve the running time of the above decision procedure, by arguing that the graph $G$ it constructs is sparse. However, ultimately that will not improve the running time of the following optimization procedure, as it searches over the $O(n^3)$ sized set of \lemref{canon}.

\begin{theorem}\thmlab{radiusapprox}
 There is an $O(n^3 \log n)$ time $(5+2\sqrt{3})$-radius-approximation algorithm for \probref{main}, when restricted to instances where $\cbs$ is a set of disjoint disks.
\end{theorem}
\begin{proof}
 By \lemref{canon}, in $O(n^3\log n)$ time we can compute an $O(n^3)$ sized set $R$ of values, such that $r_{opt}\in R$, where $r_{opt}$ is the optimal radius. So sort the values in $R$, and then binary search over them using the $(5+2\sqrt{3})$-decider of \lemref{decider}, which we denote $\mathsf{decider}(r)$. Specifically, if $\mathsf{decider}$ returns False we recurse to the right, and if it returns a solution (i.e.\ True) then we recurse on the left. Note that since our decision procedure is approximate, the values for which it returns True or for which it returns False may not be contiguous in the sorted order of $R$. Regardless, however, our binary search allows us to find a pair $r'<r$ which are consecutive in $R$ and such that $\mathsf{decider}(r')$ is False, and $\mathsf{decider}(r)$ is True. (Unless $\mathsf{decider}$ always returns True, in which case it returns the smallest value in $R$.) 
 By \lemref{decider} $\mathsf{decider}$ is a $(5+2\sqrt{3})$-decider, and thus since $\mathsf{decider}(r')$ is False by definition we have that  $r'<r_{opt}$. However, as $r'<r$ are consecutive in the sorted order of $R$ and since $r_{opt}\in R$, this implies $r_{opt}\geq r$. On the other hand, again by the definition of a $(5+2\sqrt{3})$-decider, $\mathsf{decider}(r)$ outputs a solution with radius at most $(5+2\sqrt{3})r\leq (5+2\sqrt{3})r_{opt}$, thus giving us a $(5+2\sqrt{3})$-approximation as claimed.
 
 By \lemref{canon}, computing and sorting the $O(n^3)$ values in $R$ takes $O(n^3\log n)$ time. 
  By \lemref{decider} each call to $\mathsf{decider}$ takes $O(n^{2.5})$ time, and since we are binary searching over $O(n^3)$ values, the time for all calls to $\mathsf{decider}$ is $O(n^{2.5}\log( n^{3}))=O(n^{2.5}\log n)$.
  Thus the total time is $O(n^3\log n)$ as claimed.
\end{proof}

Our focus in this paper is on the planar case, however, in \apndref{extension} we remark how the above decision procedure works in higher dimensions. 
The above optimization procedure does not immediately extend as it makes use of \lemref{canon}, however, in the appendix we informally sketch how one can approximately recover the same result.

\section{An Efficient FPTAS for Bounded k}
\seclab{bounded}

By \lemref{canon}, we can compute a set of $O(n^{3})$ points which contains a subset of size $k$ that is an optimal $k$-center solution.
Thus, for $k$ is constant, enumerating all $O(n^{3k})$ possible subsets, and taking the minimum cost solution found, yields a polynomial time algorithm. In this section, we argue that for constant $k$, we can achieve a $(1+\eps)$-radius-approximation for unit disks, whose running time depends only linearly on $n$.
Contrast this with \thmref{hardmain}, where we argued that when $k$ is not assumed to be constant, that the problem is hard to approximate for unit disks within a given constant factor.

We use the following from
Agarwal and Procopiuc \cite{ap-eaac-02}.

% In the previous section we gave a constant factor radius-approximation for disks, and by \thmref{hardmain} we know it is hard to approximate the problem for unit disks within a similar constant factor. 
% In this section we argue that when $k$ is bounded, however, we can achieve a $(1+\eps)$-approximation for unit disks.
% We use the following from
% Agarwal and Procopiuc \cite{ap-eaac-02}.

\newcommand{\kceneps}{\Algorithm{kCenter}\xspace}
\begin{theorem}[\cite{ap-eaac-02}]\thmlab{ptas}
 Given a set $P$ of $n$ points in the plane, there is an $O(n\log k)+(k/\eps)^{O(\sqrt{k})}$ time  $(1+\eps)$-radius-approximation algorithm for $k$-center, denoted $\kceneps(\eps,P)$. 
\end{theorem}

\begin{theorem}%[Proof in \apndref{ptasappendix}]
\thmlab{newptas}
 There is an $O(n\log k)+(k/\eps)^{O(k)}$ time $(1+\eps)$-radius-approximation algorithm for \probref{main}, when restricted to instances where $\cbs$ is a set of  disjoint unit disks.
\end{theorem}
\begin{proof}
 Let $P$ denote the set of center points of the disks in $\cbs$. 
 For any given set $S$ of $k$ points in the plane, let $r_P(S) = \max_{p\in P} \distX{p}{S}$ and $r_\cbs(S) = \max_{\cb\in \cbs} \distX{\cb}{S}$.  
 Observe that $r_\cbs(S) \leq r_P(S)\leq r_\cbs(S)+1$. Specifically, $r_\cbs(S)\leq r_P(S)$ since any ball (in particular one centered at a point from $S$) which contains a center point from $P$ also intersects the corresponding disk in $\cbs$. On the other hand, $r_P(S)\leq r_\cbs(S)+1$ since for any ball intersecting a disk in $\cbs$, if we increase its radius by $1$ then it will contain the center point of that disk, as $\cbs$ consists of unit disks. 
 
Let $r_{opt}$ denote the optimum radius for the given instance $\cbs,k$ of \probref{main}. We consider two cases based on the value of $r_{opt}$. First, suppose that $r_{opt}>2/\eps$. Let $S'$ denote the solution returned by $\kceneps(\eps/3,P)$. By the above inequalities and \thmref{ptas},  
 \begin{align*}
 &r_\cbs(S') \leq r_P(S') 
 \leq (1+\eps/3) \min_{S\subset \Re^2, |S|=k} r_P(S)
 \leq (1+\eps/3) (1+\min_{S\subset \Re^2, |S|=k} r_\cbs(S))\\ 
 &\!= (1+\eps/3)(1+r_{opt})
 < (1+\eps/3)(\eps r_{opt}/2+r_{opt}) 
 = (1+\eps/3)(1+\eps/2) r_{opt}
 \leq (1+\eps)r_{opt},
 \end{align*}
 where the last inequality assumed $\eps\leq 1$. Thus $S'$ is $(1+\eps)$-approximation for \probref{main}.
 
 Now suppose that $r_{opt}\leq 2/\eps$. In this case observe that for any point $x\in \Re^2$, the ball $B(x,r_{opt})$ can intersect only $O(1/\eps^2)$ disks from $\cbs$ as they are disjoint and all have radius 1. Thus any center from the optimal solution can cover at most $O(1/\eps^2)$ disks within the optimal radius, and so it must be that $n = O(k/\eps^2)$.
 
The algorithm is now straightforward. If $n \leq \gamma k/\eps^2$, for some sufficiently large constant $\gamma$, then by \lemref{canon} in $O((k/\eps^2)^3 \log (k/\eps))$ time we can compute a set $P$ of $O((k/\eps^2)^3)$ points such that $P$ contains an optimal set of $k$ centers. We try all possible subsets of $P$ of size $k$ and take the best one. There are $O((k/\eps^2)^{3k})$ such subsets, and for each subset its cost can be determined in $O(kn) = O((k/\eps)^2)$ time. Thus in this case we can compute the optimal solution in  
 $O((k/\eps)^2 \cdot (k/\eps^2)^{3k}) = (k/\eps)^{O(k)}$ time.
 
On the other hand, if $n > \gamma k/\eps^2$ then the above implies $r_{opt}> 2/\eps$. In this case it was argued above that $\kceneps(\eps/3,P)$ returns a $(1+\eps)$-approximation, and by \thmref{ptas} it does so in $O(n\log k)+(k/\eps)^{O(\sqrt{k})}$ time. 
In either case, we have a $(1+\eps)$-approximation (or better) and the total time is  
 $\max\{(k/\eps)^{O(k)}, O(n\log k)+(k/\eps)^{O(\sqrt{k})}\}$.
\end{proof}
%}

%%%%%%%%%%%%%%%%%1D%%%%%%%%%%%%%%%%%%%%%%%%%%%%
%%%%%%%%%%%%%%%%%%%%%%%%%%%%%%%%%%%%%%%%%%%%%%%

\section{One Dimensional Clustering with Neighborhoods}
\seclab{oned}
In this section we show that despite clustering with neighborhoods being hard to radius approximate within any factor in the plane, we can solve the one dimensional variant exactly in $O(n\log n)$ time, even when object intersections are allowed.
First, we argue the decision problem can be solved in linear time. Then we argue that we can use a scheme similar to that in \cite{f-pslsct-91} to search for the optimal radius. 

In one dimension, a convex object is just a closed interval. Thus we have the following one dimensional version of \probref{main}, where intersections are no longer prohibited.

\begin{problem}[One Dimensional Clustering with Neighborhoods]
  \problab{problem1d}
  Given a set $\cbs$ of $n$ closed intervals on the real line, and an integer parameter $k\geq 0$, find a set of $k$ points $S$ (called centers) which minimize the maximum distance to an interval in $\cbs$. That is, 
  \[
  S=\arg \min_{S'\subset \Re, |S'|=k} \max_{\cb\in \cbs} \distX{\cb}{S'}.
  \]
\end{problem}
%For space the details of this section are in \apndref{onedappendix}. Below is the summarizing theorem.

% \newcommand{\onedbody}{

%Here we give the details of the approach for \probref{problem1d} outlined at the beginning of \secref{oned}.
%
The following decision procedure is similar in spirit to various folklore results for interval problems in one dimension (for example, see the discussion in \cite{f-ghssfo-18} on interval stabbing). The challenge is turning this decision procedure into an efficient optimization procedure, for which as discussed below we make use of \cite{f-pslsct-91}.

We first sort the intervals in increasing order both by their left and by their right endpoints. We maintain cross links between the two sorted lists so that if we remove an interval from one list, its copy in the other list can be removed in constant time.

\begin{lemma}%[Proof in \apndref{onedappendix}]
\lemlab{feasible}
  Given an instance $\cbs,k$ of \probref{problem1d}, where the intervals have been presorted, for any query radius $r$, in $O(n)$ time one can decide whether $r \geq r_{opt}$.
  \lemlab{decision1D}
\end{lemma}
\begin{proof}
  %Let $r$ be the given query radius. 
  We build a set $S$ of centers as follows, where initially $S = \emptyset$. 
  Let $[\alpha,\beta]$ denote the interval with the leftmost right endpoint (i.e.\ $\beta$ is smallest among all intervals).
  We place a center at $\beta+r$ and add it $S$.
  Next we remove all intervals which intersect the ball $B(\beta+r,r)$.
  Note that these intersecting intervals are precisely those whose left endpoint is $\leq \beta+2r$, as this condition is clearly necessary to intersect $B(\beta+r,r)$, but also sufficient as all intervals have right end point $\geq \beta$.
  We then repeat this process until all intervals are removed. If $|S| \leq k$ we return True and otherwise we return False. 
  
  Observe that every time we place a center, we remove intervals it covers within distance $r$. Thus the final set $S$ is a set of centers of radius $r$, and so if $|S|\leq k$, then $r\geq r_{opt}$ and the algorithm correctly returns True. Moreover, we now argue that $S$ is a minimum cardinality set of centers of radius $r$, and thus if $|S|> k$ then the algorithm correctly returns False.
 Adopting notation from above, let $[\alpha,\beta]$ be the interval with leftmost right endpoint, and let $c$ be the center our algorithm places at $\beta+r$. Now in the minimum cardinality solution, there must be at least one center $c'$ within distance $r$ from $[\alpha,\beta]$, implying the location of $c'$ is $\leq \beta+r$. Thus $c'$ can only $r$-cover intervals with left endpoint $\leq \beta+2r$. However, as described above, $c$ $r$-covers all intervals with left endpoint $\leq \beta+2r$, and thus $c'$ $r$-covers a subset of those $c$ does. Conversely, the subset of intervals not $r$-covered by $c$ is a subset of those not $r$-covered by $c'$. By induction our algorithm uses the smallest possible number of centers to $r$-cover the intervals not $r$-covered by $c$, which therefore is at most the number centers the global minimum solution uses to $r$-cover the superset of intervals not $r$-covered by $c'$. 
 %Thus overall our $r$-cover was of minimum cardinality.
 Thus overall our set of centers was an $r$-cover of minimum cardinality.

  For the running time, 
  observe that determining the location of the next center takes constant time since it only depends on the leftmost right endpoint, and we assumed we have the sorted ordering of the intervals by right endpoint. Moreover, we can remove all of the intervals intersecting the $r$ radius ball at the new center in time linear in the number of intersecting intervals, since as discussed above these intersecting intervals are a prefix of the sorted ordering by left endpoint.
  As we spend constant time per interval removed, overall this is an $O(n)$ time algorithm.
\end{proof}

\lemref{decision1D} gives us a decision procedure for \probref{problem1d} which we now wish to utilize to search for the optimum radius. We use the following lemma to reduce the search space, which can be seen as a simplification of \lemref{canon} for the one dimensional case, where here we only need to consider distances from bisecting points rather than bisecting curves.

\begin{lemma}
\lemlab{radius}
 Let $\cbs$ be a set of closed intervals. Then for any value $k$, the optimal radius for the instance $\cbs,k$ of \probref{problem1d} is either 0 or $\distX{\cb}{\cb'}/2$ for some pair $\cb,\cb'\in \cbs$.
\end{lemma}
\begin{proof}
For any value $k$, let $S$ be an optimal solution with optimal radius $r_{opt}$. Consider an arbitrary center $s \in S$, and let $\cbs'$ be the subset of $\cbs$ which intersects the ball $B(s,r_{opt})$. We can assume that $|\cbs'| \geq 1$, as otherwise $B(s,r_{opt})$ does not intersect any interval and so $s$ can be thrown out.
If $|\cbs'| = 1$, then $s$ intersects only one interval, and thus without loss of generality $s$ can be placed inside this interval, i.e.\ at distance 0 from it. 
So assume $|\cbs'| > 1$, and let $\cb$ be the furthest interval from $s$ in $\cbs'$. 
As we move $s$ towards $\cb$, so long as $\cb$ remains the furthest interval from $s$ in $\cbs'$, $B(s,\distX{s}{\cb})$ will continue to intersect all intervals in $\cbs'$. If $\cb$ always remains the furthest, when $s$ eventually reaches $\cb$, its distance to $\cb$ and hence all of $\cbs'$ will be 0. Otherwise, if before we reach $\cb$, $s$ is no longer the furthest from $s$, then we must have crossed the bisector point between $\cb$ and some other interval in $\cbs'$. In this case, we can place $s$ on this bisector point and $B(s,\distX{s}{\cb})$ will intersect all intervals in $\cbs'$, and moreover $\distX{s}{\cb}\leq r_{opt}$ since $\distX{s}{\cb}$ monotonically decreased as we moved $s$ towards $\cb$. Modifying all centers in $S$ in this way thus produces a solution whose radius is $\leq r_{opt}$ and is either 0 or the distance from a bisector point to either interval in the pair it bisects.
\end{proof}

Given a set $\cbs$ of $n$ intervals, let $P(\cbs)$ denote the set of all $2n$ left and right endpoints of the intervals in $\cbs$. To find the optimal solution to an instance $\cbs,k$ of \probref{problem1d}, by \lemref{radius}, we can binary search over the interpoint distances of points in $P(\cbs)$ using our decider from \lemref{feasible}. (When we call the decider we divide the interpoint distance by two as \lemref{radius} actually tells us it is a bisector distance.)
As there are $\Theta(n^2)$ interpoint distances, naively this approach takes $O(n^2 \log n)$ time. However, \cite{f-pslsct-91} previously showed that in the abstract setting where one is given a linear time decider, and the optimal solution is an interpoint distance, one can find the optimal solution in $O(n\log n)$ time. 
This is achieved by reducing the problem to searching in an implicitly defined sorted matrix, which for completeness we now describe.
%
% This is achieved by reducing the problem to searching in an implicitly defined sorted matrix, which for completeness the details are described in \apndref{onedappendix}. 
% Below is the summarizing theorem.

%\newcommand{\onedbody}{

%Here we continue the discussion from \secref{oned}.

A matrix is said to be \textit{sorted} if the elements in every row and in every column are in nonincreasing order. Let $P=\{p_1,\ldots,p_m\}$ be a set of $m$ values on the real line, indexed in increasing order. \cite{f-pslsct-91} defines a sorted $m-1\times m-1$ matrix from $P$, containing all interpoint distance, as follows. Let $A_i = p_i - p_1$ (i.e.\ shift the points so $p_1$ is the origin). 
Observe that for any $i<j$, $p_j-p_i = A_j - A_i$. Let $M(P)$ be the $m-1 \times m-1$ matrix whose $ij$-th entry is $A_{m+1-i} - A_j$.
It is easy to see that $M(P)$ is a sorted matrix, and as $P$ is indexed in sorted order we have an $O(n)$ space implicit representation of $M(P)$ where each entry can be computed in constant time. 

We now reproduce the description of the procedure \textit{MSEARCH}, presented in \cite{f-pslsct-91} (which combines ideas from \cite{f-oatp-91,fj-fpcgsgds-83,fj-gsrsm-84}). The input is a set of sorted matrices, a stopping count $c$, and a searching range $(\lambda_1, \lambda_2)$ such that $\lambda_2$ is feasible and $\lambda_1$ is not, where initially we set $(\lambda_1, \lambda_2) = (0,\infty)$.
\textit{MSEARCH} produces a sequence of values one at a time to be tested for feasibility, where the result of each test allows us to discard some elements in the set of matrices. 
If a value $\lambda \notin (\lambda_1, \lambda_2)$ is produced, then it does not need to be tested. If $\lambda$ is feasible, then $\lambda_2$ is reset to $\lambda$, otherwise $\lambda_1$ is reset to $\lambda$.
\textit{MSEARCH} stops once the number of matrix elements remaining is no larger than the stopping count.

\begin{lemma}[\cite{f-pslsct-91}, Theorem 2.1]
  \lemlab{msearch}
  Let $\mathcal{M}$ be a set of $N$ sorted matrices $\{M_1, M_2,\allowbreak \ldots, M_N \}$ in which matrix $M_j$ is of dimension $m_j \times n_j, m_j \leq n_j$, and $\sum^N_{j=1} m_j = m$.
  Let $c \geq 0$. The number of feasibility tests needed by \textit{MSEARCH} to discard all but at most $c$ of the elements is $O(\{\max \{\log \max_j\{ n_j\}, \log (m/(c+1)) \})$, and the total time of \textit{MSEARCH} exclusive of feasibility tests is $O(\sum^N_{j=1} m_j \log (2n_j / m_j))$.
\end{lemma}

In our case we have a single sorted matrix which we exhaustively search for the optimum by setting $c=0$. (Note \cite{f-pslsct-91} allowed for multiple sorted matrices as their input was a tree which they decomposed into multiple paths.)  
Thus we have the following simplified corollary.

\begin{corollary}
  Given an $m\times m$ sorted matrix $M$,  
  the number of feasibility tests needed by \textit{MSEARCH} to find the optimum is 
  $O(\log m)$, and the total time of \textit{MSEARCH} exclusive of feasibility tests is $O(m)$.
\end{corollary}

Thus if we set $M=M(P(\cbs))$ (and hence $m+1 = 2n$ in the above corollary), 
then \textit{MSEARCH} with our linear time decision procedure from \lemref{feasible} gives
%the following.
\thmref{onedmain}.
%}

\begin{theorem}\thmlab{onedmain}
  \probref{problem1d} can be solved in $O(n\log n)$ time, where $n=|\cbs|$.
\end{theorem}

%\newpage

%\bibliographystyle{alpha}
\bibliographystyle{plain}
\bibliography{refs}%

%%%%%%%%%%%%%%%%%%%%%%%%%%%%%%%%%%%%%%%%%
%%%%%%%%%%%%%%%%%%%%%%%%%%%%%%%%%%%%%%%%%
\newpage
\appendix

\section{Extending the Disk Approximation to Higher Dimensions}\apndlab{extension}
Here we informally remark how the results in \secref{balls} can be extended to higher dimensions. To do so we need to extend \lemref{arc}, \lemref{decider}, and \thmref{radiusapprox}. 

To extend \lemref{arc}, we consider the same setup as in the proof in two dimensions. Namely let $x$, $y$, and $z$ be the centers of the disjoint $r$ radius balls (which are now in $\Re^d$), and let $s$ be some fourth point. Observe that the centers $x$, $y$, and $z$ define a two dimensional plane, and moreover the intersection of their respective balls with this plane is a set of three disjoint $r$ radius disks in this plane. Let $s'$ denote the orthogonal projection of $s$ into this plane.  We now have the same two dimensional setup as in \lemref{arc}, and thus the argument there applies so long as we can argue the distance from $s'$ to any one of the three disks in the plane lower bounds the distance from $s$ in $\Re^d$ to any one of the three balls. To argue this it suffices to observe that $||x-s'||^2\leq ||x-s||^2$ (and similarly for $y$ and $z$). Specifically, suppose we rotate space so that this plane corresponds to the first two coordinate axes. Then, 
$||x-s'||^2 = (x_1-s_1')^2+(x_2-s_2')^2 = (x_1-s_1)^2+(x_2-s_2)^2 \leq ||x-s||^2$.

A careful read of \lemref{decider} reveals that in fact the same proof works in $\Re^d$, if we just change the word ``disk'' to ``ball''. Thus at this point we have a $(5+2\sqrt{3})$-decider that works in $\Re^d$. \thmref{radiusapprox} showed how to turn this into a optimization procedure in the plane. Specifically, in the proof we search using our approximate decider over a set of values given by \lemref{canon}. The issue with extending the optimization procedure to higher dimensions is that \lemref{canon} no longer applies. 
That said, one can find an approximate set of radii to search over, and this is implied by the proof of \lemref{decider}.

Specifically, consider the proof of \lemref{decider} when $r=r_{opt}$. A set of $k$ centers $S$ is constructed which is the disjoint union of three types of centers. Namely, type $S_1$, which are center points of input disks (i.e.\ input balls in $\Re^d$), type $S_2$ which are arbitrary points in a disk, and type $S_3$ which are midpoints between two disks. Every input disk is assigned to exactly one of these types of centers. A disk is only assigned to a center of type $S_2$ if that center lies in the disk, i.e.\ is at distance zero.
Disks assigned to $S_1$ or $S_3$ centers are at distance at most $(5+2\sqrt{3})r_{opt}$ from their respective center.
So let $x$ denote the largest distance from a disk to its assigned center. Then $x\leq (5+2\sqrt{3})r_{opt}$. On the other hand $x\geq r_{opt}$, since $r_{opt}$ is the minimum disk to center distance under an optimal assignment to an optimal set of $k$ centers, whereas $x$ is determined by some set of $k$ centers $S$ and potentially a non-optimal assignment of disks to centers in $S$. 

Let $R$ be the set containing all distances between an input disk and a disk center point (i.e.\ $S_1$ type distances), the value zero (i.e.\ $S_2$ types distances), and all distances between the midpoint of two disks to one of the two disks (i.e.\ $S_3$ type distances). Then by the above $x\in R$, and thus $R$ contains a value which is constant factor approximation to $r_{opt}$. There are a quadratic number of  values in $R$, which we can compute, and then binary search over using our decider. 
This yields an $O(1)$ approximation to $r_{opt}$, though not necessarily a $(5+2\sqrt{3})$-approximation (as both the decider and the value $x$ were approximate). However, one can turn it into a $(5+2\sqrt{3}+\eps)$-approximation with $O(1/\eps)$ additional calls to our $(5+2\sqrt{3})$-decider, by using standard techniques. (Namely, given any constant spread interval $[z,cz]$ containing the optimum, one uses the decider to exponential search over all values $z(1+\eps)^i$ in this interval.) 
Thus, for any constant $\eps>0$, there is a polynomial time $(5+2\sqrt{3}+\eps)$-approximation algorithm which works for input balls in $\Re^d$, as the proof of \lemref{decider} and the above discussion generalizes from disks to balls in $\Re^d$.

% \section{One Dimensional Approach and Proofs}
% \apndlab{onedappendix}
% 
% \onedbody

\end{document}